\documentclass[submission,copyright]{eptcs}

\providecommand{\event}{GandALF 2022}
\usepackage{stmaryrd}  
\usepackage{tikz}
\usetikzlibrary{arrows,positioning}

\usepackage{bussproofs}
\usepackage{enumerate}
\usepackage{cite} 
\usepackage{caption}
\usepackage{subcaption}
\usepackage{xspace,centernot}

\newcommand*{\nat}{\ensuremath{\mathbb{N}}}

\newcommand{\logicize}[1]{{\ensuremath{\mathbf{#1}}}\xspace}
\renewcommand{\k}{\logicize{K}}
\renewcommand{\d}{\logicize{D}}
\renewcommand{\t}{\logicize{T}}
\renewcommand{\b}{\logicize{B}}
\newcommand{\kf}{\logicize{K4}}
\newcommand{\df}{\logicize{D4}}
\newcommand{\sr}{\logicize{S4}}
\newcommand{\sv}{\logicize{S5}}
\newcommand{\kv}{\logicize{K5}}

\newcommand\change[1]{#1}
\newcommand{\LG}{\logicize{L}}

\usepackage{complexity}

\renewcommand{\M}{\mathcal{M}}

\usepackage{amsmath}
\usepackage{amsthm}
\usepackage{amsfonts}
\usepackage{amssymb}

\usepackage{multicol}
\usepackage{thmtools}

\newtheorem{theorem}{Theorem}
\newtheorem{definition}{Definition}
\newtheorem{corollary}[theorem]{Corollary}
\newtheorem{lemma}[theorem]{Lemma}

\newtheorem{proposition}[theorem]{Proposition}

\theoremstyle{remark}
\newtheorem{remark}{Remark}
\newtheorem{example}{Example}

\theoremstyle{definition}
\newtheorem{translation}{Translation}

\newcommand{\mycap}[1]{\text{\sc#1}}

\newcommand{\trueset}[1]{{\left\llbracket #1 \right\rrbracket}}
\newcommand{\diam}[1]{\langle #1 \rangle}

\newcommand{\stbox}[1]{[ #1 ]}
\newcommand{\al}{\alpha}
\newcommand{\alb}{\beta}

\newcommand{\proc}{\ensuremath{W}}  

\newcommand{\impl}{\rightarrow}
\newcommand{\tran}{\ensuremath{R}}
\newcommand{\trana}{\ensuremath{R_\al}}

\newcommand{\false}{\mathtt{ff}}
\newcommand{\true}{\mathtt{tt}}

\newcommand{\act}{\mycap{Ag}}

\newcommand{\ML}{modal logic}
\newcommand{\NI}{Negative Introspection}

\newcommand{\mn}{\ensuremath{\mu}}
\newcommand{\mx}{\ensuremath{\nu}}

\newcommand{\transl}[3]{\ensuremath{\mathsf{F}_{#1}^{#2}({#3})}}

\newcommand{\nStt}[2]{.#1\langle#2\rangle}
\newcommand{\fix}{\ensuremath{\mathsf{fx}}}
\newcommand{\form}{\ensuremath{\Phi}}

\newcommand{\resp}{\emph{resp.}~}

\newcommand{\sub}{\ensuremath{\mathsf{sub}}}
\newcommand{\subc}{\ensuremath{\overline{\mathsf{sub}}}}

\newcommand{\ax}{\ensuremath{\mathsf{ax}}}

\newcommand{\qedd}{\hfill$\blacksquare$}

\author{Luca Aceto 
\institute{Department of Computer Science\\ Reykjavik University \\
Reykjavik, Iceland}
\institute{Gran Sasso Science Institute \\ L'Aquila, Italy}
\email{luca@ru.is} 
\and Antonis Achilleos  
\institute{Department of Computer Science\\ Reykjavik University \\
Reykjavik, Iceland}
\email{\quad antonios@ru.is }
\and Elli Anastasiadi
\institute{Department of Computer Science\\ Reykjavik University \\
Reykjavik, Iceland}
\email{\quad elli19@ru.is }
\and Adrian Francalanza
\institute{Department of Computer Science, ICT \\ University of Malta \\ Msida, Malta
}
\email{adrian.francalanza@um.edu.mt}
\and Anna Ingolfsdottir
\institute{Department of Computer Science\\ Reykjavik University \\ Reykjavik, Iceland}
\email{annai@ru.is}
}
\def\titlerunning{Complexity through Translations for Modal Logic with Recursion}
\def\authorrunning{Aceto et al.}

\title{Complexity through Translations for Modal Logic with Recursion\thanks{ 
	This work has been funded by the projects ``Open Problems in the Equational Logic of Processes (OPEL)'' (grant no.~196050), ``Epistemic Logic for Distributed Runtime Monitoring'' (grant no.~184940), ``Mode(l)s of Verification and Monitorability'' (MoVeMent) (grant no~217987) of the Icelandic Research Fund, and the project ``Runtime and Equational Verification of Concurrent Programs'' (ReVoCoP) (grant 222021) of the Reykjavik University Research Fund.}
}

\begin{document}
	
	\maketitle
	
\begin{abstract} 

This paper studies the complexity of classical modal logics and of their extension with fixed-point operators, using 
translations to transfer results across logics. In particular, we show several complexity results
for 
multi-agent logics via translations to and from the $\mu$-calculus and modal logic, which allow us to transfer known upper and lower bounds. We also use these translations to introduce a terminating tableau system for the logics we study, based on Kozen's tableau for the $\mu$-calculus, and the one of Fitting and Massacci for modal logic.
\end{abstract}

%

\section{Introduction}
\label{sec:intro}

We introduce a family of multi-modal logics with fixed-point operators that are interpreted on restricted classes of Kripke models.
One can consider these logics as extensions of the usual multi-agent logics of knowledge and belief  \cite{Fagin1995ReasoningAboutKnowledge} by adding recursion to their syntax or of the $\mu$-calculus \cite{Kozen1983a} by interpreting formulas on different classes of frames and thus giving an epistemic interpretation to the modalities.
We define \emph{translations} between these logics, and we demonstrate how 
one can rely on these translations to
prove  finite-model theorems, complexity bounds, and tableau termination for each logic in the family.

Modal logic comes in several variations \cite{blackburn2006handbook}. 
Some of these, such as multi-modal logics of knowledge and belief  \cite{Fagin1995ReasoningAboutKnowledge}, are of particular interest to Epistemology and other application areas. 
Semantically, the classical modal logics used in epistemic (but also other) contexts result from imposing certain restrictions on their models.
On the other hand, the modal $\mu$-calculus \cite{Kozen1983a} can be seen as an extension of the smallest normal modal logic \k\ with greatest and least fixed-point operators, $\mx X$ and $\mn X$ respectively.
We explore the situation where one allows both recursion (\emph{i.e.} fixed-point) operators in a multi-modal language and imposes restrictions on the semantic models.

We are interested in the complexity of satisfiability for the resulting logics.
Satisfiability for the 
$\mu$-calculus is known to be \EXP-complete \cite{Kozen1983a}, while for the modal logics between $\k$ and $\sv$ the problem is \PSPACE-complete or \NP-complete, depending on whether they have \NI\ \cite{ladnermodcomp,Halpern2007Characterizing}.
In the multi-modal case, satisfiability for those modal logics becomes \PSPACE-complete, and is \EXP-complete with the addition of a common knowledge operator \cite{Halpern1992}.

\change{
There is plenty of relevant work on 
the $\mu$-calculus on restricted frames, 
mainly in its single-agent form.
Alberucci and Facchini examine the alternation hierarchy of the $\mu$-calculus over reflexive, symmetric, and transitive frames in \cite{alberucci_facchini_2009}.
D'Agostino and Lenzi have studied the $\mu$-calculus on different classes of frames in great detail.  
In \cite{DAGOSTINO20104273transitive}, they reduce the $\mu$-calculus over finite transitive frames to first-order logic.
In \cite{Dagostino2013S5}, they prove that $\sv^\mu$-satisfiability is \NP-complete, and that the two-agent version of $\sv^\mu$ does not have the finite model property.
In \cite{Dagostino2015modal}, they consider finite symmetric frames, and they prove that $\b^\mu$-satisfiability is in 2\EXP, and \EXP-hard.
They also examine planar frames in \cite{DAGOSTINO201840planar}, where they show that the alternation hierarchy of the $\mu$-calculus over planar frames is infinite.
}

Our primary method of proving complexity results is through translations to and from the multi-modal $\mu$-calculus.
We show that we can use surprisingly simple translations from modal logics without recursion to the base modal logic $\k_n$, reproving the \PSPACE\ upper bound for these logics \change{(Theorem \ref{thm:no-fixpoint-translations-work} and Corollary \ref{cor:modalupper})}.
These translations and our constructions to prove their correctness do not generally transfer to the corresponding logics with recursion.
We present translations from specific logics to the $\mu$-calculus and back, and we discuss the remaining open cases.
We discover, through the properties of our translations, that several behaviors induced on the transitions do not affect the complexity of the satisfiability problem.
\change{%
	As a result, we prove that all logics with axioms among $D$, $T$, and $4$, and the  least-fixed-point fragments of logics that also have $B$,  
have their satisfiability in \EXP, and a matching lower bound for the logics with axioms from $D,T,B$ (Corollaries \ref{cor:mu-upper1} and \ref{cor:mu-upper2}).%
}
Finally, we present tableaux for the discussed logics, based on the ones by Kozen for the $\mu$-calculus \cite{Kozen1983a}, and by Fitting and Massacci for modal logic \cite{fitting1972tableau,Massacci1994}.
We give tableau-termination conditions for every logic with a finite model property \change{(Theorem \ref{thm:tableaux})}.

The addition of recursive operators to \ML\ increases expressiveness.
An important example is that of \emph{common knowledge} or \emph{common belief}, which can be expressed with a greatest fixed-point thus:
 $\mx X.(\varphi \land \bigwedge_\al [\al] X)$. 
But the combination of epistemic logics and fixed-points
can potentially express more interesting epistemic concepts.
For instance, the formula
$\mn X. \bigvee
_{\al 
}(
[\al]\varphi \vee [\al] X)$, 
in the context of a belief interpretation, 
can be thought to claim that there is a rumour of $\varphi$.
It would be interesting to see what other meaningful sentences of epistemic interest one can express using recursion. 
Furthermore, the family of logics we consider allows each agent to behave according to a different logic. 
This flexibility allows one to mix different interpretations of modalities, such as a temporal interpretation for one agent and an epistemic interpretation for another.
Such logics can even resemble hyper-logics \cite{Hyperproperties} if a set of agents represents different streams, and combinations of epistemic and temporal or hyper-logics have recently been used to express safety and privacy properties of systems \cite{EpistLogicRV}.

The paper is organized as follows. Section \ref{sec:back} gives the necessary background and an overview of current results.
Section \ref{sec:translations} defines a class of translations that provide us with several upper and lower bounds, and identifies  conditions under which they can be composed.
In Section \ref{sec:multi} we finally give tableaux for our multi-modal logics with recursion. 
We conclude in Section \ref{sec:conclusion} with a set of open questions and directions.
Omitted proofs can be found in the 
full version of the paper \cite{CTMLR2022onlinepreprint}.

\section{Definitions and Background}
\label{sec:back} 
This section introduces the logics that we study and the necessary background on the complexity of  \ML\ and the $\mu$-calculus.

\subsection{The Multi-Modal Logics with Recursion}
We start by defining the syntax of the logics.
\begin{definition}
	We consider formulas  constructed from the following grammar:
	\begin{align*}
	\varphi,\psi \in L
	::&=
	p &&|~~
	\neg p &&|~~
	\true                     &&|~~\false      &	&|~~X     
	&&|~~\varphi\land\psi            &&|~~\varphi\lor\psi    \\
	&|~~\langle\al\rangle\varphi &&|~~[\al]\varphi    
	&&|~~\mn X.\varphi              &&|~~\mx X.\varphi                                                     ,
	\end{align*}
	where $X$ comes from a countable set of logical (or fixed-point) variables, \change{$\mycap{LVar}$}, $\al$ from a finite set of agents, $\act$, and $p$ from a finite 
	set of propositional variables\change{, $\mycap{PVar}$}. 
	When $\act = \{\al \}$, 
	$\Box \varphi$ 
	stands for
	$[\al]\varphi$, and $\Diamond \varphi$ 
	for
	$\diam{\al}\varphi$. 
	We also write $[A]\varphi$ to mean $\displaystyle\bigwedge_{\al \in A} [\al]\varphi$ and $\diam{A}\varphi$ 
	\change{for} 
	$\displaystyle\bigvee_{\al \in A} \diam{\al}\varphi$.
\end{definition}

A formula is closed when every occurrence of a variable $X$  is in the scope of recursive operator $\mx X$ or $\mn X$. 
Henceforth we 
consider
only 
closed formulas, unless we specify otherwise.

\change{Moreover, for recursion-free closed formulas we associate the notion of \textit{modal depth}, which 
	is the nesting depth
	of the modal operators%
	\footnote{\change{The modal depth of recursive formulas can be either zero, or infinite. However, this is not relevant for the spectrum of this work.}}. The modal depth of $\varphi$ is defined inductively as:
\begin{itemize}
\item $md(p) = md(\neg p) = md(\true) = md(\false) = 0$, where $p \in \mycap{PVar}$,
\item $md(\varphi \lor \psi) = md(\varphi \land \psi) = max(md(\varphi),md(\psi))$, and
\item $md(\stbox{a}\varphi) = md(\diam{a}\varphi)=  1 + md(\varphi)$, where $a \in \act$.
\end{itemize}}
We assume that in formulas, each recursion variable $X$ appears in a unique fixed-point formula $\fix(X)$, which is either of the form $\mn X.\varphi$ or $\mx X.\varphi$.
If $\fix(X)$ is a least-fixed-point (\resp greatest-fixed-point) formula, then $X$ is called a least-fixed-point (\resp greatest-fixed-point) variable.
We can define a partial order on fixed-point variables, such that 
$X\leq Y$ iff $\fix(X)$ is a subformula of $\fix(Y)$, and $X < Y$ when $X \leq Y$ and $X \neq Y$.
If $X$ is $\leq$-minimal among the free variables of $\varphi$, then 
we define the \emph{closure} of $\varphi$ to be 
$cl(\varphi)=cl(\varphi[\fix(X)/X])$,
 where $\varphi[\psi/X]$ is the usual substitution operation, and if $\varphi$ is closed, then 
 $cl(\varphi)=\varphi$.
 
 We define $\sub(\varphi)$ as the set of subformulas of $\varphi$, and  $|\varphi|=|\sub(\varphi)|$ is bounded by the length of $\varphi$ as a string of symbols.
 Negation, $\neg \varphi$, and implication, $\varphi \impl \psi$, can be defined in the usual way.
 Then, we define $\subc(\varphi) = \sub(\varphi) \cup \{ \neg \psi \in L \mid \psi \in \sub(\varphi) \}$.
 


\paragraph{Semantics}

We interpret formulas on 
the states of a
\emph{Kripke model}.
%
A Kripke model, or simply 
model, 
is
a quadruple
\change {$M=$} $( \proc,\tran,V)$
where $\proc$ is a nonempty set of states, 
${\tran}\subseteq\proc\times\act \times \proc$ is a transition relation, and $V:\proc\to 2$\change{$^{\mycap{PVar}}$} determines on which states a propositional variable is true. 
$( \proc,\tran)$ is called a \emph{frame}.
We usually write $(u,v)\in \trana$ or $u \trana v$ instead of $(u,\al,v)\in \tran$, or $u \tran v$, \change{when $\act$ is a singleton $\{\al\}$}.

Formulas are evaluated in the context of 
an \emph{environment} $\rho:\mycap{LVar}\to 2^\proc$, which gives values to the logical variables. 
For an environment $\rho$, variable $X$, and set $S \subseteq \proc$, we write $\rho[X \mapsto S]$ for the environment that maps $X$ to $S$ and all $Y \neq X$ to $\rho(Y)$.
The semantics for 
our
formulas is given through a function $ \trueset{ -  }_\M$, defined in Table \ref{table:semantics}.
\begin{table}
\begin{align*}
\noalign{\noindent
$ \trueset{\true, \rho }= \proc, 
\hfill 
 \trueset{\false,\rho }=\emptyset ,
\hfill
 \trueset{ p, \rho } = \{s \mid p\in V(s) \}
,
\hfill
 \trueset{ \neg p, \rho } = \proc {\setminus}  \trueset{ p, \rho } 
,
$}
 \trueset{[\al]\varphi,\rho }&=\left\{s~\big|~ \forall t. \ sR_{\al}t\text{ implies } t\in \trueset{\varphi,\rho }\right\},
& \trueset{\varphi_1{\land}\varphi_2, \rho }&= \trueset{\varphi_1,\rho }\cap \trueset{\varphi_2,\rho }, 
\\
 \trueset{\langle\al\rangle\varphi,\rho }&=\left\{s~\big|~ \exists t. \ sR_{\al}t\text{ and } t\in \trueset{\varphi,\rho }\right\},
 & 
 \trueset{\varphi_1{\lor}\varphi_2, \rho }&= \trueset{\varphi_1,\rho }\cup \trueset{\varphi_2,\rho } , 
\\
 \trueset{\mn X.\varphi,\rho }&=\bigcap\left\{S~\big|~ S \supseteq  \trueset{\varphi,\rho[X\mapsto S] }\right\},
 &
 \trueset{ X, \rho }& = \rho(X),
\\
 \trueset{\mx X.\varphi,\rho }&=\bigcup\left\{S~\big|~ S \subseteq  \trueset{\varphi, \rho[X\mapsto S] }\right\} .
\end{align*}
\caption{Semantics of modal formulas on a model $\M = (W,R,V)$. We omit $\M$ from the 
	notation.}
\label{table:semantics}
\end{table}
The semantics of $\neg \varphi$ 
are constructed 
as usual, where $ \trueset{ \neg X, \rho }_\M = \proc {\setminus} \rho(X)$.

We sometimes use  
$\M,s \models_\rho \varphi$ for $s \in \trueset{\varphi,\rho}_\M$, and 
as the environment 
has no effect on the semantics of a closed formula $\varphi$, 
we 
often 
drop it from the notation 
 and 
write 
$\M,s \models \varphi$ or $s \in \trueset{\varphi}_\M$.
%
If $\M,s \models \varphi$, we say that $\varphi$ is true, or satisfied, in $s$. 
When the particular model does not matter, or is clear from the context, we may omit it. 


Depending on how we 
further 
restrict our 
syntax 
and 
the model, we can describe several logics.
Without further restrictions, the resulting logic 
is the $\mu$-calculus \cite{Kozen1983a}. The max-fragment (resp. min-fragment) of the $\mu$-calculus is the fragment that only allows 
the $\mx X$ (resp. the $\mn X$) recursive operator.
If 
$|\act| = k$
and we allow no recursive operators (or recursion variables), then we have the basic modal logic $\k_k$ (or $\k$, if $k=1$), and further restrictions on the frames can result in a wide variety of modal logics (see \cite{MLBlackburnRijkeVenema}). 
We give names to the following frame conditions, or frame constraints, for the case where $\act=\{\al\}$.
These conditions correspond to the usual axioms for normal modal logics --- see \cite{blackburn2006handbook,MLBlackburnRijkeVenema,Fagin1995ReasoningAboutKnowledge}, which we will revisit in Section \ref{sec:translations}.
\begin{multicols}{2}
\begin{description}
	\item[$D$:] 
	$\tran$ is serial: 
	$\forall s.\exists t. s  \tran t $;
	\item[$T$:] $\tran$ is reflexive: 
	$\forall s. s  \tran s $;
	\item[$B$:] $\tran$ is symmetric: 
	$\forall s,t. {(s \tran t   \Rightarrow  t \tran s)}$;
	\item[$4$:] $\tran$ is transitive: 
	$\forall s,t,r. {(s \tran t \tran r  \Rightarrow  s \tran r)}$;
	\item[$5$:] $\tran$ is euclidean: 
	$\forall s,t,r.$ if $s \tran t$ and $s \tran r$, \\  then $t \tran r$.
\end{description}
\end{multicols}

We consider modal logics that are interpreted over models that satisfy a combination of these constraints for each agent. 
$D$, which we call Consistency, is a special case of $T$, called Factivity. Constraint $4$ is 
Positive Introspection and $5$ is called Negative Introspection.\footnote{These are names for properties or axioms of a logic. When we refer to these conditions as conditions of a frame or model, we may refer to them with the name of the corresponding relation condition: seriality, reflexivity, symmetry, transitivity, and euclidicity.}
Given a logic $\LG$ and constraint $c$, $\LG+c$ is the logic that is interpreted over all models with frames that satisfy  all the constraints of $\LG$ and $c$.
The name of a single-agent logic is a combination of the constraints that apply to its frames, including $K$, if the constraints are among $4$ and $5$.
Therefore, 
logic $\d$ is $\k + D$, $\t$ is $\k+T$, $\b$ is $\k+B$, $\kf = \k + 4$, 
$\df = \k + D + 4 = \d + 4$, 
and so on.
We use the special names $\sr$ for $\logicize{T4}$ and $\sv$ for $\logicize{T45}$.
We define a (multi-agent) logic $\LG$ on $\act$ as a map from agents to single-agent logics.
$\LG$ is interpreted on Kripke models of the form $(\proc,\tran,V)$, where for every $\al\in\act$, $(\proc,\trana)$ is a frame for $\LG(\al)$.

For a logic $\LG$, 
$\LG^\mu$ is the logic that results from $\LG$ after we allow recursive operators in the syntax --- in case they were not allowed in \LG.
Furthermore, if for every $\al \in \act$, $\LG(\al)$ is the same single-agent logic ${\LG}$, we write $\LG$ as ${\LG}_k$, where $|\act|=k$.
Therefore, the $\mu$-calculus is $\k^\mu_k$.

From now on, unless we explicitly say otherwise, by a logic, we mean one of the logics we 
have 
defined above.
%
%
We call a formula satisfiable  for a logic $\LG$, if it is satisfied in some  state of a model  for $\LG$. 

\begin{example}
	For a formula $\varphi$, we  define $Inv(\varphi)= \mx X.(\varphi \land [\act]X)$. $Inv(\varphi)$ asserts that $\varphi$ is true in \emph{all} reachable states, or, alternatively, it can be read as an assertion that $\varphi$ is common knowledge.
	We  dually define 
	$Eve(\varphi)= \mn X.(\varphi \lor \diam{\act}X)$, which asserts that $\varphi$ is true in \emph{some} reachable state. 
\end{example}

\subsection{Known Results}
\label{sec:known-results}
For logic $\LG$, the satisfiability problem for $\LG$, or $\LG$-satisfiability is the problem that asks, given a formula $\varphi$, if $\varphi$ is satisfiable. Similarly, the model checking problem for $\LG$ asks if $\varphi$ is true at a given state of a given \change{finite} model. 

 Ladner \cite{ladnermodcomp} established the classical result of \PSPACE-completeness for the satisfiability of \k, \t, \d, \kf, \df, and \sr\ and \NP-completeness for the satisfiability of \sv.  Halpern and R\^{e}go later characterized the \NP--\PSPACE\ gap  for one-action logics by the presence or absence of Negative Introspection \cite{Halpern2007Characterizing}, resulting in Theorem \ref{thm:ladhalp}. Later, Rybakov and Shkatov \cite{RybakovShkatovBComp} proved the \PSPACE-completeness of \b\ and \t\b.
 For formulas with fixed-point operators, 
 D'Agostino and Lenzi in \cite{Dagostino2013S5} show that satisfiability for single-agent logics with constraint $5$ is also \NP-complete.

\begin{theorem}[\cite{ladnermodcomp,Halpern2007Characterizing,RybakovShkatovBComp}]\label{thm:ladhalp}
	If $\LG\in \{\k, \t, \d, \b, \t\b, \kf, \df, \sr\}$, then 
	$\LG$-satisfiability is \PSPACE-complete; and $\LG+5$-satisfiability and $(\LG+5)^\mu$-satisfiability is \NP-complete.
\end{theorem}

\begin{theorem}[\cite{Halpern1992}]\label{thm:halpKkcomp}
	If $k>1$ and $\LG$ has a combination of constraints from $D, T, 4, 5$ and no recursive operators, then 
	$\LG_k$-satisfiability is \PSPACE-complete.
\end{theorem}

\begin{remark}\label{remark:Halpern-and-Moses-should-have-proved-more}
	Note that Halpern and Moses in \cite{Halpern1992} prove these bounds for the cases of $\k_k, \t_k, \sr_k, \logicize{KD45}_k,$ and $\sv_k$ only.
	Similarly, D'Agostino and Lenzi in \cite{Dagostino2013S5} only prove the \NP-completeness of satisfiability for $\sv^\mu$. However, it is not hard to see that their respective methods also work for the rest of the logics of Theorems \ref{thm:ladhalp} and \ref{thm:halpKkcomp}.
	\qedd
\end{remark}


\begin{theorem}[\cite{Kozen1983a}]\label{prp:muCalc-sat}
	The satisfiability problem for the $\mu$-calculus 
	is \EXP-complete.
\end{theorem}


\begin{theorem}[\cite{emerson2001model}]\label{thm:mu-calc-MC}
	The model checking problem for the $\mu$-calculus is in $\NP \cap \coNP$.\footnote{In fact, the problem is known to be in $\UP \cap \coUP$ \cite{jurdzinski1998deciding}.}
\end{theorem}

Finally we have the following initial known results about the complexity of satisfiability, when we have recursive operators. Theorems \ref{prp:fragmantsMu} and \ref{prp:EXP-hard-k2} have already been observed in \cite{AcetoAFI20}. 

\begin{theorem}\label{prp:fragmantsMu}
	The satisfiability problem for the min- and max-fragments of the $\mu$-calculus is \EXP-complete, even 
	when $|\act| = 1$. 
\end{theorem}

\begin{proof}[Proof sketch]
	It is known that satisfiability for the min- and max-fragments of the $\mu$-calculus (on one or more action) is \EXP-complete. It is in \EXP\ due to Theorem \ref{prp:muCalc-sat}, 
	and 
	these fragments  suffice \cite{Pratt_1981} to describe the PDL formula that is constructed by the reduction used in \cite{fischer1979propositional} to prove \EXP-hardness for PDL. Therefore, that reduction can be adjusted to prove that 
	satisfiability for the min- and max-fragments of the $\mu$-calculus  is \EXP-complete.
\end{proof}

It is not hard to express in logics with both frame constraints and recursion operators that formula $\varphi$ is common knowledge, with formula $\mx X.\varphi \land [\act]X$. Since validity for $\LG_k$ with common knowledge (and without recursive operators) and $k > 1$ is \EXP-complete \cite{Halpern1992}\footnote{Similarly to Remark \ref{remark:Halpern-and-Moses-should-have-proved-more}, \cite{Halpern1992} does not explicitly cover all these cases, but the techniques can be adjusted.}, $\LG_k^\mu$ 
\change{is}
$\EXP$-hard.

\begin{proposition}\label{prp:EXP-hard-k2}
	Satisfiability for $\LG_k^\mu$, where $k>1$, is \EXP-hard.
\end{proposition}

%

\section{Complexity through Translations}
\label{sec:translations}

In this section, 
we examine 
\LG-satisfiability.
We  use formula translations to reduce the satisfiability of one logic to the satisfiability of another.
We investigate the properties of these translations and how they compose with each other, and we achieve complexity bounds for several  logics.

In the context of this paper,
a formula translation from logic $\LG_1$ to logic $\LG_2$ is a mapping $f$ on formulas such that each formula $\varphi$ is $\LG_1$ -satisfiable if and only if $f(\varphi)$ is $\LG_2$ -satisfiable.
We only consider translations that can be computed in polynomial time, and therefore, our translations are polynomial-time reductions, transfering complexity bounds between logics.

According to Theorem \ref{prp:muCalc-sat}, $\k^\mu_k$-satisfiability is \EXP-complete, and therefore for each logic $\LG$, we aim to connect $\k^\mu_k$ and $\LG$ via a sequence of translations in either direction, to prove complexity bounds for $\LG$-satisfiability.

%
%

\subsection{Translating Towards $\k_k$}
We begin by presenting translations from logics with more to logics with fewer frame conditions.
To this end, we study how taking the closure of a frame under one condition affects any other frame conditions.

\subsubsection{Composing Frame Conditions}
\label{subsec:compositionality}


We now discuss how the conditions for frames affect each other. 
For example, to construct a transitive frame, one can take the transitive closure of a possibly non-transitive frame. 
The resulting frame will satisfy condition $4$. 
As we see, taking the closure of a frame under condition $x$ may affect whether that frame maintains condition $y$, depending on $x$ and $y$.
In the following we observe that one can apply the frame closures in certain orders that preserve the properties one aquires with each application.



Let $F = (W,R)$ be a frame, $\al \in A \subseteq \act$, and $x$ a frame restriction among $T, B, 4, 5$.
Then, $\overline{R_\al}^{x}$ is the closure of $R_\al$ under $x$, 
$\overline{R}^{x,A}$ is defined by $\overline{R}^{x,A}_\beta = \overline{R_\beta}^{x}$, if $\beta \in A$, and $\overline{R}^{x,A}_\beta = R_\beta$, otherwise.
Then, $\overline{F}^{x,A} = (W,\overline{R}^{x,A})$. We make the following observation.



\begin{lemma}\label{lem:conditionsarepreserved}
	Let $x$ be a frame restriction among $D, T, B, 4, 5$, and $y$ a frame restriction among $T, B, 4, 5$, such that $(x,y) \neq (4,B), (5,T), (5,B)$.
	Then, for every frame $F$ that satisfies $x$, $\overline{F}^y$ also satisfies $x$.
%
	%
	%
\end{lemma}

According to Lemma \ref{lem:conditionsarepreserved}, frame conditions are preserved as seen in Figure \ref{fig:frame_hierarchy}. 
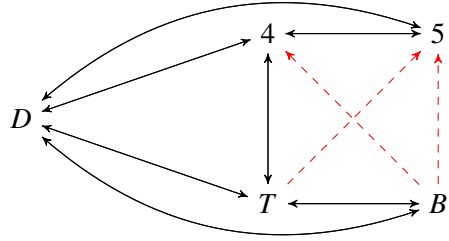
\begin{figure}
	\begin{center}
		\begin{tikzpicture}[scale=0.65,->,>=stealth',node distance=1.6cm, 
			main node/.style={circle,draw,font=\large\bfseries}]
			
			\node (c) { };
			\node (4) [above left of=c] {$4$};
			\node (t) [below left of=c] {$T$};
			\node (b) [below right of=c] {$B$};
			\node (5) [above right of=c] {$5$};
			\node (d) [left= 4cm of c] {$D$};
			
			\path 
			(t) edge (4)
			edge [red,dashed] (5)
			edge (b)
			edge (d)
			

			(d) edge (t) 
			edge [bend right] (b)
			edge (4)
			edge [bend left] (5)
			
			(4) 
			edge (d)
			edge (5)
			edge (t)
			(b) edge [red,dashed] (4)
			edge [red,dashed] (5)
			edge (t)
			edge [bend left] (d)
			(5) edge (4)
			edge [bend right](d)
			;
		\end{tikzpicture}
	\end{center}
	\caption{The frame property hierarchy}
	\label{fig:frame_hierarchy}
\end{figure}
%
%
%
In Figure \ref{fig:frame_hierarchy}, an arrow from $x$ to $y$ indicates that  property $x$ is preserved though the closure of a frame under $y$. 
Dotted red arrows indicate one-way arrows.
For convenience, we define $\overline{F}^D = (W,\overline{R}^D)$, where $\overline{R}^D = R \cup \{ (a,a) \in W^2 \mid \centernot{\exists} (a,b) \in R \}$.

\begin{remark}
	We note that, in general, not all frame conditions are preserved through all closures under another condition. For example, the accessibility relation $\{(a,b),(b,b)\}$ is euclidean, but its reflexive closure $\{(a,b),(b,b),(a,a)\}$ is not.
\end{remark}

There is at least one linear ordering of the frame conditions $D, T, B, 4, 5$, such that all preceding conditions are preserved 
by
closures under the following conditions. We call such an order a closure-preserving order.
We use the linear order $D, T, B, 4, 5$ in the rest of the paper.


%

\subsubsection{Modal Logics}
\label{subsec:trans-ml}

We start with translations that map logics 
without recursive operators to logics with fewer constraints.
As mentioned in Subsection \ref{sec:known-results}, all of the logics  $\LG\in \{\k, \t, \d, \kf, \df, \sr\}$ and $\LG+5$ with one agent have known completeness results, and the complexity of modal logic is well-studied for multi-agent modal logics as well. 
The missing cases are very few and concern the combination of frame conditions (other than $5$) as well as the multi-agent case. However we take this 
opportunity to present an 
intuitive introduction to our general translation method.
%
%
In fact,
the translations that we use for logics without recursion are surprisingly straightforward.
Each frame condition that we introduced in Section \ref{sec:back} is associated with an axiom for modal logic, such that whenever a model has the condition, every substitution instance of the axiom is satisfied in all worlds of the model (see \cite{blackburn2006handbook,MLBlackburnRijkeVenema,Fagin1995ReasoningAboutKnowledge}). We give  for each frame condition $x$ and agent $\al$, the axiom $\mathsf{ax}^x_\al$: 
\begin{multicols}{3}
	\begin{description}
		\item[$\ax^D_\al$: ] $\diam{\al}\true$
		\item[$\ax^T_\al$: ] $[\al]p\impl p$
		\item[$\ax^B_\al$: ] $\diam{\al}[\al]p \impl p $
		\item[$\ax^4_\al$: ] $[\al]p \impl [\al][\al]p$
		\item[$\ax^5_\al$: ] $\diam{\al}[\al]p\impl [\al]p$
	\end{description}
\end{multicols}

For each formula $\varphi$ and $d \geq 0$, let $Inv_d(\varphi) = \bigwedge_{
	i \leq d}[\act]^i\varphi$.
Our first translations are straightforwardly defined from the above axioms.

%

\begin{translation}[One-step Translation]\label{translation:ml}
	Let $A\subseteq \act$ and let $x$ be one of the frame conditions.
	For every formula $\varphi$, let $d=md(\varphi)$ 
	if $x\neq 4$, and $d=md(\varphi)|\varphi|$, if $x=4$. We define:
	$$ \transl{A}{x}{\varphi} = \varphi \wedge Inv_d\Bigg(\bigwedge_{\substack{\psi \in \subc(\varphi) \\ \al \in A}} \ax^x_\al[
	\psi 
	/ p]\Bigg).
	$$
\end{translation}




\begin{theorem}\label{thm:no-fixpoint-translations-work}
	Let $A\subseteq \act$, $x$ be one of the frame conditions, and let $\LG_1, \LG_2$ be logics without recursion operators, such that 
	$\LG_1(\al) = \LG_2(\al) + x$ when $\al \in A$, and $\LG_2(\al)$ otherwise, and $\LG_2(\al)$ only includes frame conditions that precede $x$ in the fixed order of frame conditions.
	Then, $\varphi$ is $\LG_1$-satisfiable if and only if $ \transl{A}{x}{\varphi}$ is $\LG_2$-satisfiable.
\end{theorem}
We present here a short proof sketch of this theorem.
\begin{proof}[Proof sketch]
The proof of the ``only if'' direction is straightforward, as for any agent $\al$ with frame condition $x$, $\ax_\al^x$ is valid in $\LG_1$.
Thus the translation holds on any $\LG_1$-model that satisfies $\varphi$.

The other, and more involved direction 
requires the construction of an $\LG_1$-model for $\varphi$ from an $\LG_2$-model for $\transl{A}{x}{\varphi}$.
%
We make use of the observation that no
modal logic formula can describe a model at depth more than the  constant $d$.
Therefore, we use the unfolding of the $\LG_2$-model, to keep track  of the path that one takes to reach a certain state, and that path's length.
We then carefully reapply the necessary closures on the accessibility relations and we use induction on its subformulas, to prove that $\varphi$ is true in the constructed model.
We do this as a separate case for each axiom.
It is worth noting that we pay special care for the case of $x=4$ to account for the fact that the translation does not allow the modal depth of the relevant subformulas to decrease with each transition during the induction; and that we needed to include the negations of subformulas of $\varphi$ in the conjunction of Translation \ref{translation:ml}, only for the case of $x=5$.
\end{proof}

\begin{corollary}\label{cor:modalupper}
	The satisfiability problem for 
	every logic without fixed-point operators
	is in \PSPACE.
\end{corollary}



\subsubsection{Modal Logics with Recursion}
\label{sec:transl_mu_cal}

In the remainder of this section we will modify our translations and proof technique, in order to lift our results to 
logics with fixed-point operators.
It is not clear whether the translations of Subsection \ref{subsec:trans-ml} can be extended straightforwardly in the case of logics with recursion, by using unbounded invariance $Inv$, instead of the bounded  $Inv_d$.
\begin{example}
	Let 
$\varphi_f=$ \change{$\mn$}$ X. \Box X$, which requires all paths in the model to be finite, and thus it 
is
not satisfiable
in reflexive frames. 
In Subsection \ref{subsec:trans-ml}, to translate formulas from reflexive models, we did not need to add the negations of subformulas as conjuncts.
In this case, such a translation would give
$$\varphi_t := \varphi_f \land Inv((\Box \varphi_f \impl \varphi_f) \land (\Box\Box \varphi_f \impl \Box\varphi_f)
).
$$
 \change{Indeed, on reflexive frames,} the formulas $\Box \varphi_f \impl \varphi_f$ and $\Box\Box \varphi_f \impl \Box\varphi_f$ are valid, and therefore $\varphi_t$ is equivalent to $\varphi_f$, which is \k-satisfiable.
%
This was not an issue 
in Subsection \ref{subsec:trans-ml}, as
the finiteness of the paths in a model 
cannot be expressed without recursion.

One would then naturally wonder whether conjoining over $\subc(\varphi_f)$ in the translation
 would make a difference.
The answer is affirmative, as the tranlation 
	$$ \varphi_f \wedge Inv\Bigg(\bigwedge_{\psi \in \subc(\varphi_f) } 
\Box\psi \impl \psi 
\Bigg)
$$
would then yield a formula that is not satisfiable.
However, our constructions would not work to prove that such a translation preserves satisfiability.
For example, consider $\mn X.\Box(p \impl(r \land( q \impl X)))$, whose translation is satisfied on a pointed model that satisfies at the same time $p$ and $q$.
We invite the reader to verify the details.
\end{example}

The only case \change{where} the 
approach that we used for the logics without recursion can be applied 
is for 
the case of seriality (condition $D$), as $Inv(\diam{\al}\true)$ directly ensures the seriality of a model.

\begin{translation}
	$$ 
	\transl{A}{D^{\mu}}{\varphi}
	= \varphi  \wedge Inv\big(\bigwedge_{\al \in A}\diam{\al} \true 
	\big).$$  
\end{translation}

\begin{theorem}\label{thm:mu_calc_serial_transl}
	Let $A\subseteq \act$ and $|\act|=k$, 
	and let $\LG$ be a logic, such that 
	$\LG(\al) = \d$ when $\al \in A$, and $\k$ otherwise.
	Then, $\varphi$ is $\LG$-satisfiable if and only if 
	$\transl{A}{D^{\mu}}{\varphi}$ is $\k_k^\mu$-satisfiable.
\end{theorem}

For the cases of reflexivity and transitivity, our simple translations substitute the modal subformulas of a  formula to implicitly enforce the corresponding condition.

\begin{translation}
The operation $\transl{A}{T^{\mu}}{-}$ is defined to be such that
\begin{itemize}
	\item 
$
\transl{A}{T^{\mu}}{[\al]\varphi}
= \stbox{\al} \transl{A}{T^{\mu}}{\varphi} \wedge \transl{A}{T^{\mu}}{\varphi}$;
\item 
$
\transl{A}{T^{\mu}}{\diam{\al}\varphi}
= \diam{\al} \transl{A}{T^{\mu}}{\varphi} \vee \transl{A}{T^{\mu}}{\varphi}$;
\item and it commutes with all other operations.
\end{itemize}
\end{translation}

\begin{theorem}\label{thm:mu_calc_reflex_translation}
	Let $\emptyset \neq A\subseteq \act$, 
	and let $\LG_1, \LG_2$ be logics, such that 
	$\LG_1(\al) = \LG_2(\al) + T$ when $\al \in A$, and $\LG_2(\al)$ otherwise, and $\LG_2(\al)$ at most includes frame condition $D$.
	Then, $\varphi$ is $\LG_1$-satisfiable if and only if 
	$\transl{A}{T^{\mu}}{\varphi}$ is $\LG_2$-satisfiable.
\end{theorem}
\begin{proof}[Proof sketch]
	The ``only if'' \change{direction} is proven by taking the appropriate reflexive closure and showing by induction that the subformulas of $\varphi$ are preserved.
\end{proof}

\begin{translation}
	The operation $\transl{A}{4^{\mu}}{-}$ is defined to be such that
	\begin{itemize}
		\item $\transl{A}{4^{\mu}}{[\al] \psi} = Inv([\al](\transl{A}{4^{\mu}}{\psi})$, 
		\item $\transl{A}{4^{\mu}}{\diam{\al} \psi} = Eve(\diam{\al}(\transl{A}{4^{\mu}}{\psi})$\change{,}
		
		\item $\transl{A}{4^{\mu}}-$ commutes with all other operations.
		
	\end{itemize}
\end{translation}


\begin{theorem}\label{thm:mu_calc_transitive_transl}
	Let $\emptyset \neq A\subseteq \act$, 
	and let $\LG_1, \LG_2$ be logics, such that 
	$\LG_1(\al) = \LG_2(\al) + 4$ when $\al \in A$, and $\LG_2(\al)$ otherwise, and $\LG_2(\al)$ at most includes frame conditions $D, T, B$.
	Then, $\varphi$ is $\LG_1$-satisfiable if and only if 
	$\transl{A}{4^{\mu}}{\varphi}$ is $\LG_2$-satisfiable.
\end{theorem}

\begin{proof}
	If $\transl{A}{4^{\mu}}{\varphi}$ is satisfied in a model $M = (W,R,V)$, let $M' = (W,R^+,V)$, where $R^+_\al$ is the transitive closure of $R_\al$, if $\al \in A$, and $R^+_\al = R_\al$, otherwise.
	It is now not hard to use induction on $\psi$ to show that for every (possibly open) subformula $\psi$ of $\varphi$, for every environment $\rho$,
	$\trueset{\transl{A}{4^{\mu}}{\psi},\rho}_\M = \trueset{\psi,\rho}_{\M'}$.
	The other direction is more straightforward.
\end{proof}

In order to produce a similar translation for symmetric frames, we needed to use a more intricate type of construction. Moreover, 
we only prove the correctness of the following translation for formulas without least-fixed-point operators.

\begin{translation}\label{transl:mu_calc_symm}
The operation $\transl{A}{\b^{\mu}}{-}$ is defined 
as
$$\transl{A}{\b^{\mu}}{\varphi} = \varphi \land Inv\big(\stbox{\al} \diam{\al} p \wedge \bigwedge_{\psi \in \subc(\varphi)} (\psi \rightarrow [\al][\al](p \impl \psi))\big)\change{,}$$ \change{where $p$ is a new propositional variable, not occurring in $\varphi$.}
%
%
\end{translation}
%

\begin{theorem}\label{thm:mu_calc_symm_translation}
	Let $\emptyset \neq A\subseteq \act$, 
	and let $\LG_1, \LG_2$ be logics, such that 
	$\LG_1(\al) = \LG_2(\al) + \b$ when $\al \in A$, and $\LG_2(\al)$ otherwise, and $\LG_2(\al)$ at most includes frame conditions $D,T$.
	Then, a  formula $\varphi$ that has no $\mn X$ operators is $\LG_1$-satisfiable if and only if 
	$\transl{A}{\b^{\mu}}{\varphi}$ is $\LG_2$-satisfiable.
\end{theorem}

\begin{remark}\label{remark:finitemodelB5}
	A translation for euclidean frames and for the full syntax on symmetric frames would need different approaches.
	 D'Agostino and Lenzi show in \cite{Dagostino2013S5} that $\sv_2^\mu$ does not have a finite model property, and their result can be easily extended to any logic $\LG$ with fixed-point operators, where there are at least two distinct agents $\al,\alb$, such that $\LG(\al)$ and $\LG(\alb)$ have constraint $B$ or $5$.
	 Therefore, 
	 as our constructions for the translations to $\k^\mu_k$ guarantee the finite model property to the corresponding logics, they do not apply to multimodal logics with $B$ or $5$.
\end{remark}
\subsection{Embedding $\k_n^\mu$}
In this subsection, we present translations from logics with fewer frame conditions to ones with more conditions.
This will allow us to prove \EXP-completeness in the following subsection.
	Let
$p,q$ be distinguished propositional variables that do not appear in 
our formulas.
We let $\vec{p}$ range over $p$, $\neg p$, $p \wedge q$, $p \wedge \neg q$,
and
$\neg p\wedge q$%
.

\begin{definition}[function $next$]
	$next(p \wedge q) = p \wedge \neg q$, $next(p\wedge \neg q) = \neg p \wedge q$, and $next(\neg p \wedge q) = p \wedge q$; and $next(p) = \neg p$ and $next(\neg p) = p$. 
\end{definition}


We use a uniform translation from $\k^\mu_k$ to any logic with a combination of conditions $D,T,B$.

\begin{translation}\label{transl:lower}
The operation 
$\transl{A}{\k^\mu}{-}$
on formulas is defined such that:
\begin{itemize}
	\item $\transl{A}{\k^\mu}{\diam{\al} \psi} = \bigwedge_{\vec{p}}(\vec{p} \rightarrow \diam{\al} (next(\vec{p}) \wedge \transl{A}{\k^\mu}{\psi}))
	$, if 
	$\al \in A$;
	\item $\transl{A}{\k^\mu}{[\al] \psi}=  \bigwedge_{\vec{p}} (\vec{p} \rightarrow [\al] (next(\vec{p}) \rightarrow \transl{A}{\k^\mu}{\psi}))
	$, if $\al \in A$;
	\item $\transl{A}{\k^\mu}{-}$ commutes with all other operations.
\end{itemize}
\end{translation}

We note that there are simpler translations for the cases of logics with only $D$ ot $T$ as a constraint, but the $\transl{A}{\k^\mu}{-}$ is uniform for all the logics that we consider in this subsection.

\begin{theorem}\label{thm:lower-translations}
	Let $\emptyset \neq A\subseteq \act$, $|\act|=k$, 
	and let $\LG$ be 
	such that 
	$\LG(\al)$ 
	includes only frame conditions from $\d,\t,\b$
	when $\al \in A$, and $\LG(\al)=\k$ otherwise.
	Then, $\varphi$ is $\k_k^{\mu}$-satisfiable if and only if 
	$\transl{A}{\k^{\mu}}{\varphi}$
	is $\LG$-satisfiable.
\end{theorem}
The proof of Theorem \ref{thm:lower-translations} can be found in \cite{CTMLR2022onlinepreprint}.
It is worth noting that the ``if'' direction uses the symmetric closure to construct a new model, while the ``only if'' direction requires the introduction of new \change{states} that behave as each original state in the model.

%




\subsection{Complexity results}

	\change{We observe that our translations all result in formulas of size at most linear with respect to the original. The exceptions are Translations \ref{translation:ml} and \ref{transl:mu_calc_symm}, which have a quadratic cost.}

\begin{corollary}\label{cor:mu-upper1}
		If  $\LG$ only has frame conditions $D,T$, then its satisfiability problem is \EXP-complete;
		if $\LG$ only has frame conditions $D,T,4$, then its satisfiability problem is in \EXP.
\end{corollary}

\begin{proof}
	Immediately from Theorems \ref{thm:mu_calc_serial_transl}, \ref{thm:mu_calc_reflex_translation}, \ref{thm:mu_calc_transitive_transl}, and \ref{thm:lower-translations}.
\end{proof}

\begin{corollary}\label{cor:mu-upper2}
	If $\LG$ only has frame conditions $D,T,B$, then 
	\begin{enumerate}
		\item 
		$\LG$-satisfiability is \EXP-hard;  and 
		\item 
		the restriction of $\LG$-satisfiability on formulas without $\mn X$ operators is \EXP-complete.
	\end{enumerate}
\end{corollary}
\begin{proof}
Immediately from Theorems \ref{thm:mu_calc_serial_transl}, \ref{thm:mu_calc_reflex_translation}, \ref{thm:mu_calc_symm_translation}, and \ref{thm:lower-translations}.
\end{proof}

\section{Tableaux for $\LG_k^\mu$}
\label{sec:multi}

We give a sound and complete tableau system for logic $\LG$.
Furthermore, if $\LG$ has a finite model property, then we give terminating conditions for its tableau. 
The system that we give in this section is based on Kozen's tableaux for the $\mu$-calculus \cite{Kozen1983a} and the tableaux of Fitting \cite{fitting1972tableau} and Massacci \cite{Massacci1994} for \ML.
We can use Kozen's finite model theorem \cite{Kozen1983a} to help us ensure
the termination of the tableau for some of these logics.

\begin{theorem}[\cite{Kozen1983a}]
	\label{thm:smallKmodel}
	There is a computable $\kappa:\nat \to \nat$, such that every $\logicize{K}_k^\mu$-satisfiable formula $\varphi$ is satisfied in a model with at most $\kappa(|\varphi|)$ states.\footnote{The tableau in 
		\cite{Kozen1983a} yields an upper bound of 
		$2^{2^{O(n^3)}}$ for $\kappa_0(n)$, 
		but that bound is not useful to 
		obtain
		a ``good'' decision procedure. The purpose of this section is not to establish any good upper bound for satisfiability testing, which is done in Section \ref{sec:translations}.}
\end{theorem}

\begin{corollary}\label{cor:smallLmodel}
	If $\LG$ only has frame conditions $D, T, 4$, then 
	there is a computable $\kappa:\nat \to \nat$, such that  every $\LG$-satisfiable formula $\varphi$ is satisfied in a model with at most $\kappa(|\varphi|)$ states.
\end{corollary}
\begin{proof}
	Immediately, from Theorems \ref{thm:smallKmodel}, \ref{thm:mu_calc_serial_transl}, \ref{thm:mu_calc_reflex_translation}, and \ref{thm:mu_calc_transitive_transl}, and Lemma \ref{lem:conditionsarepreserved}.
\end{proof}

\begin{remark}
	We note that not all modal logics with recursion have a finite model property -- see Remark \ref{remark:finitemodelB5}.
\end{remark}

\change{Intuitively, a tableau attempts to build a model that satisfies the given formula. 
When it needs to consider two possible cases, it branches, and thus it may generate several branches. Each branch that satisfies certain consistency conditions, which we define below, represents a corresponding model.}

Our tableaux use \emph{prefixed formulas}, that is, formulas of the form $\sigma~\varphi$, where $\sigma \in (\act\times L)^*$ and $\varphi \in L$; $\sigma$ is the prefix of $\varphi$ in that case\change{, and we say that $\varphi$ is prefixed by $\sigma$.
We note that we separate the elements of $\sigma$ with a dot.}
We say that the prefix $\sigma$ is $\al$-flat when $\al$ has axiom $5$ and $\sigma = \sigma'\nStt{\al}{\psi}$ for some $\psi$.
\change{Each prefix possibly represents a state in a corresponding model, and a prefixed formula $\sigma~\varphi$ declares that $\varphi$ is satisfied in the state represented by $\sigma$. As we will see below, the prefixes from $(\act\times L)^*$ allow us to keep track of the diamond formula that generates a prefix through the tableau rules. For 
	agents with condition $5$, this allows us to restrict the generation of new prefixes and avoid certain redundancies, due to the similarity of euclidean binary relations to equivalence relations \cite{nagle_thomason_1985,Halpern2007Characterizing}.}

\begin{table}[]
	\centering 
	\begin{tabular}{c c c c } 
		\AxiomC{$\sigma~\pi X.\varphi$}
		\RightLabel{(\textsf{fix})}
		\UnaryInfC{$\sigma~\varphi$}
		\DisplayProof &
		\AxiomC{$\sigma~X$}
		\RightLabel{(\textsf{X})}
		\UnaryInfC{$\sigma~\fix(X)$}
		\DisplayProof 
		&
		\AxiomC{$\sigma~\varphi \lor \psi$}
		\RightLabel{(\textsf{or})}
		\UnaryInfC{$\sigma~\varphi\mid\sigma~\psi$}
		\DisplayProof &		
		\AxiomC{$\sigma~\varphi \land \psi$}
		\RightLabel{(\textsf{and})}
		\UnaryInfC{$\sigma~\varphi$}
		\noLine
		\UnaryInfC{$\sigma~\psi$}
		\DisplayProof 
		\\[4ex]
%
	\AxiomC{$\sigma~[\al]\varphi$}
	\RightLabel{(\textsf{B})}
	\UnaryInfC{$\sigma\nStt{\al}{\psi}~\varphi$}
	\DisplayProof~&
	\AxiomC{$\sigma~\diam{\al}\varphi$}
	\RightLabel{(\textsf{D})}
	\UnaryInfC{$\sigma\nStt{\al}{\varphi}~\varphi$}
	\DisplayProof~&
	\AxiomC{$\sigma~[\al]\varphi$}
	\RightLabel{(\textsf{d})}
	\UnaryInfC{$\sigma\nStt{\al}{\varphi}~\varphi$}
	\DisplayProof~&
	\AxiomC{$\sigma~[\al]\varphi$}
	\RightLabel{(\textsf{4})}
	\UnaryInfC{$\sigma\nStt{\al}{\psi}~[\al]\varphi$}
	\DisplayProof
	\\
\end{tabular}
\\[2ex]
\noindent
where, for rules (\textsf{B}) and (\textsf{4}), 
\change{$\sigma\nStt{\al}{\psi}$} has already appeared in the branch;
and for (\textsf{D}), $\sigma$ is not $\al$-flat.
\\[2ex]%
\noindent
\begin{tabular}{c c c l l} 
	\AxiomC{$\sigma\nStt{\al}{\psi}~[\al]\varphi$}
	\RightLabel{(\textsf{B5})}
	\UnaryInfC{$\sigma~[\al]\varphi$}
\DisplayProof &
\AxiomC{$\sigma\nStt{\al}{\psi}~\diam{\al}\varphi$}
\RightLabel{(\textsf{D5})}
\UnaryInfC{$\sigma\nStt{\al}{\psi}\nStt{\al}{\varphi}~\varphi$}
\DisplayProof &
\AxiomC{$\sigma\nStt{\al}{\psi}~[\al]\varphi$}
\RightLabel{(\textsf{b})}
\UnaryInfC{$\sigma~\varphi$}
\DisplayProof &
\AxiomC{$\sigma~[\al]\varphi$}
\RightLabel{(\textsf{t})}
\UnaryInfC{$\sigma~\varphi$}
\DisplayProof
\\[4ex]
\AxiomC{$\sigma\nStt{\al}{\psi}~[\al]\varphi$}
\RightLabel{(\textsf{B55})}
\UnaryInfC{$\sigma\nStt{\al}{\psi'}~[\al]\varphi$}
\DisplayProof &
\AxiomC{$\sigma\nStt{\al}{\psi}\nStt{\al}{\psi'}~\diam{\al}\varphi$}
\RightLabel{(\textsf{D55})}
\UnaryInfC{$\sigma\nStt{\al}{\psi}\nStt{\al}{\varphi}~\varphi$}
\DisplayProof  & 
\AxiomC{$\sigma\nStt{\al}{\psi}~[\al]\varphi$}
\RightLabel{(\textsf{b4})}
\UnaryInfC{$\sigma~[\al]\varphi$}
\DisplayProof
	\\
\end{tabular}
\\[2ex]
\noindent
where, 
for rule (\textsf{B55}),  
$\sigma\nStt{\al}{\psi'}$ has already appeared in the branch;
for rule (\textsf{D5}),  $\sigma$ is not $\al$-flat, and 
$\sigma~\diam{\al}\varphi$ does not appear in the branch;
for rule (\textsf{D55}),  
$\sigma~\diam{\al}\varphi$ does not appear in the branch.
%
\caption{The tableau rules for $\logicize{L}=\logicize{L}^\mu_n$}
\label{tab:tableau}
\end{table}
The tableau rules that we use appear in Table \ref{tab:tableau}. These include  fixed-point and propositional rules, as well as rules that deal with modalities. Depending on the logic that each agent $\al$ is based on, a different set of rules applies for $\al$: for rule (\textsf{d}), $\logicize{L}(\al)$ must have condition $D$; for rule (\textsf{t}), $\logicize{L}(\al)$ must have condition $T$; for rule (\textsf{4}), $\logicize{L}(\al)$ must have condition $4$; for rule (\textsf{B5}), (\textsf{D5}), and (\textsf{D55}), $\logicize{L}(\al)$ must have condition $5$; for (\textsf{b}) $\logicize{L}(\al)$ must have condition $B$; and for (\textsf{b4}) $\logicize{L}(\al)$ must have both $B$ and $4$.
Rule (\textsf{or}) is the only rule that splits the current tableau branch into two.
A tableau branch is propositionally closed when $\sigma~\false$ or both $\sigma~p$ and $\sigma~\neg p$ appear in the branch for some prefix $\sigma$.
For each prefix $\sigma$ that appears in a \change{fixed} tableau \change{branch}, let $\form(\sigma)$ be the set of formulas prefixed by $\sigma$ \change{in that branch}. 
We use the notation $\sigma \prec \sigma'$ to mean that $\sigma'=\sigma.\sigma''$ for some $\sigma''$, 
in which case
$\sigma$ is an ancestor of $\sigma'$.

We define the relation $\xrightarrow{X}$ on prefixed formulas in a tableau \change{branch} as $\chi_1 \xrightarrow{X} \chi_2$, if $\frac{\chi_1}{\chi_2}$ is a tableau rule and $\chi_1$ is not of the form $\sigma~Y$, where $X<Y$; then, $\xrightarrow{X}^+$ is the transitive closure of $\xrightarrow{X}$ and $\xrightarrow{X}^*$ is its reflexive and transitive closure.
We can also extend this relation to prefixes, so that $\sigma \xrightarrow{X} \sigma'$, if and only if $\sigma~\psi \xrightarrow{X} \sigma'~\psi'$, for some $\psi\in\form(\sigma)$ and $\psi'\in\form(\sigma')$.
If in a branch
there is a
$\xrightarrow{X}$-sequence where $X$ is a least fixed-point and appears infinitely often,
then the branch is called fixed-point-closed.
A branch is closed when it is either fixed-point-closed or propositionally closed; if it is not closed, then it is called open.

Now, assume that there is a 
$\kappa:\nat \to \nat$, such that  every $\LG$-satisfiable formula $\varphi$ is satisfied in a model with at most $\kappa(|\varphi|)$ states.
An open tableau branch is called (\emph{resp. locally}) \emph{maximal} when 
all tableau rules (\emph{resp. the tableau rules that do not produce new prefixes}) have been applied.
A branch is called \emph{sufficient} for $\varphi$ when it is locally maximal and for every $\sigma~\psi$ in the branch, for which a rule can be applied and has not been applied to $\sigma~\psi$, 
$|\sigma| > |\act|\cdot\kappa(|\varphi|)^{|\varphi|^2}\cdot 2^{2|\varphi|+1}$.
%
A tableau is called maximal when all of its open branches are maximal, and closed when all of its branches are closed.
It is called sufficiently closed for $\varphi$ if it is propositionally closed, or for some least fixed-point variable $X$, it 
has a $\xrightarrow{X}$-path, where $X$ appears at least $\kappa(|\varphi|)+1$ times.
%
A sufficient branch for $\varphi$ that is not sufficiently closed is called sufficiently open for $\varphi$.

A tableau for $\varphi$ starts from $\varepsilon~\varphi$ and is built using the tableau rules of Table \ref{tab:tableau}.
A tableau proof for $\varphi$ is a closed tableau for the negation of $\varphi$.

\begin{theorem}[Soundness, Completeness, and Termination of $\LG_k^\mu$-Tableaux]
	\label{thm:tableaux}
	From the following, the first two are equivalent for any formula $\varphi\in L$ and any logic $\LG$. Furthermore, if 
	there is a 
	$\kappa:\nat \to \nat$, such that  every $\LG$-satisfiable formula $\varphi$ is satisfied in a model with at most $\kappa(|\varphi|)$ states, then all the following are equivallent.
	\begin{enumerate}
	\itemsep0em 
		\item $\varphi$ has a maximal $\logicize{L}$-tableau with an open branch;
		\item $\varphi$ is $\logicize{L}$-satisfiable; and
		\item $\varphi$ has an $\logicize{L}$-tableau with a sufficiently open branch for $\varphi$.
	\end{enumerate}
\end{theorem}

\begin{proof}[Proof sketch]
	The direction from 1 to 2 uses the usual model construction, but where one needs to take into account the fixed-point formulas; the direction from 2 to 3 uses techniques and results from \cite{Kozen1983a,zielonka1998infinite}, including Corollary \ref{cor:smallLmodel}; and the direction from 3 to 1 shows how to detect appropriate parts of the branch to repeat until we safely get a maximal branch.
\end{proof}

\begin{corollary}
	$\logicize{L}$-tableaux are sound and complete for $\logicize{L}$.
\end{corollary}
\newcommand{\closed}{$\mathtt{x}$}

\begin{example}
	\change{Let $\act = \{a,b\}$ and $\LG$ be a logic, such that 
		$\LG(a) = \k^\mu$ 
		and $\LG(b) = \kv^\mu$. 
		Let
	\begin{align*}
		\varphi_1 &= (p \land \diam{a}p) \land \mn X.(\neg p \vee [a]X)& &and& &\varphi_2 = \diam{b}p \land \mn X.([b]\neg p \vee [b])X.
	\end{align*}
	As we see in Figure \ref{tab:tableauex}, the tableau for $\varphi_1$ produces an open branch, while the one for $\varphi_2$ has all of its branches closed, the leftmost one due to an infinite $\xrightarrow{X}$-sequence.
}
	
	\alwaysRootAtTop    
	\def\defaultHypSeparation{\hskip .1in}
		

\begin{figure}
\begin{subfigure}[b]{0.5\textwidth}
$
	\AxiomC{$a\diam{p}~[a]X$}
	\AxiomC{\closed}
	\noLine
	\UnaryInfC{$a\diam{p}~\neg p$}
	\BinaryInfC{$a\diam{p}~\neg p \vee [a]X$}\RightLabel{\scriptsize(\textsf{fix})}
	\UnaryInfC{$a\diam{p}~\mn X.(\neg p \vee [a]X)$}\RightLabel{\scriptsize(\textsf{X})}
	\UnaryInfC{$a\diam{p}~X$}\RightLabel{\scriptsize(\textsf{B})}
	\UnaryInfC{$a\diam{p}~p$}\RightLabel{\scriptsize(\textsf{D})}
	\UnaryInfC{$\varepsilon~[a]X$}
	\AxiomC{\closed}
	\noLine
	\UnaryInfC{$\varepsilon~\neg p$}
	\BinaryInfC{$\varepsilon~\neg p \vee [a]X$}\RightLabel{\scriptsize(\textsf{fix})}
	\UnaryInfC{$\varepsilon~\diam{a}p$}
	\noLine 
	\UnaryInfC{$\varepsilon~p$}
	\UnaryInfC{$\varepsilon~p \land \diam{a}p$}
	\noLine 
	\UnaryInfC{$\varepsilon~\mn X.(\neg p \vee [a]X)$}
	\UnaryInfC{$\varepsilon~(p \land \diam{a}p) \land \mn X.(\neg p \vee [a]X)$}
	\DisplayProof $
\end{subfigure}
\begin{subfigure}[b]{0.5\textwidth}
$ \AxiomC{\vdots}\RightLabel{\scriptsize(\textsf{X})}
	\UnaryInfC{$b\diam{p}~X$}\RightLabel{\scriptsize(\textsf{B})}
	\UnaryInfC{$\varepsilon~[b]X$}\RightLabel{\scriptsize(\textsf{B5})}
	\UnaryInfC{$b\diam{p}~[b]X$} 
	\AxiomC{\closed}\noLine 
	\UnaryInfC{$b\diam{p}~\neg p$}\RightLabel{\scriptsize(\textsf{B})}
	\UnaryInfC{$\varepsilon~[b]\neg p$}\RightLabel{\scriptsize(\textsf{B5})}
	\UnaryInfC{$b\diam{p}~[b]\neg p$} 
	\BinaryInfC{$b\diam{p}~[b]\neg p \vee [b]X$}\RightLabel{\scriptsize(\textsf{fix})}
	\UnaryInfC{$b\diam{p}~\mn X.([b]\neg p \vee [b]X)$}\RightLabel{\scriptsize(\textsf{X})}
	\UnaryInfC{$b\diam{p}~X$}\RightLabel{\scriptsize(\textsf{B})}
	\UnaryInfC{$\varepsilon~[b]X$}
	\AxiomC{\closed}
	\noLine
	\UnaryInfC{$b\diam{p}~\neg p$}\RightLabel{\scriptsize(\textsf{B})}
	\UnaryInfC{$\varepsilon~[b]\neg p$}
	\BinaryInfC{$\varepsilon~[b]\neg p \vee [b]X$}\RightLabel{\scriptsize(\textsf{fix})}
	\UnaryInfC{$b\diam{p}~p$}\RightLabel{\scriptsize(\textsf{D})}
	\UnaryInfC{$\varepsilon~\diam{b}p$}
	\noLine 
	\UnaryInfC{$\varepsilon~\mn X.([b]\neg p \vee [b]X)$}
	\UnaryInfC{$\varepsilon~\diam{b}p \land \mn X.([b]\neg p \vee [b]X)$}
	\DisplayProof $
\end{subfigure}
\caption{\change{Tableaux for $\varphi_1$ and $\varphi_2$. The dots represent that the tableau keeps repeating as from the identical node above. The \texttt{x} mark represents a propositionally closed branch.}}
\label{tab:tableauex}
\end{figure}
\end{example}


\section{Conclusions}
\label{sec:conclusion}

We studied multi-modal logics with recursion. These logics mix the frame conditions from epistemic modal logic, and the recursion of the $\mu$-calculus.
We gave simple translations among these logics that connect their satisfiability problems.
This allowed us to offer complexity bounds for
satisfiability and to prove certain finite model results. We also presented a sound and complete tableau
that has termination guarantees, conditional on a logic’s finite model property.

\paragraph{Conjectures and Future Work}
We currently do not posses full translations for the cases of symmetric and euclidean frames. 
What is interesting is that we also do not have a counterexample to prove that the translations that we already have, as well as other attempts, are \emph{not} correct. 
In the case of symmetric frames, we have managed to prove that our construction works 
for formulas without least-fixed-point operators.
A translation for euclidean frames and for the full syntax on symmetric frames
is left as future work. 
We know that we cannot use the same model constructions that preserve the finiteness of the model as in Subection \ref{sec:transl_mu_cal} (see Remark \ref{remark:finitemodelB5}).
%


We do not prove the finite model property on all logics. We note that although it is known  that \change{logics} with recursion with at least two agents with either $B$ or $5$ do not have this property (see \ref{remark:finitemodelB5}, \cite{Dagostino2013S5}), the situation is unclear if there is only one such agent.

We further conjecture that it is not possible to prove \EXP-completenes for all the single-agent cases. Specifically, 
we expect $\kf^\mu$-satisfiability to be in \PSPACE,  similarly to how 
$\kv^\mu$-satisfiability is in \NP \cite{Dagostino2013S5}.
As such, we do not expect Translation \ref{transl:lower} to be correct for these cases.

\change{The model checking problem for the $\mu$-calculus is an important open problem.
The problem does not depend on the frame restrictions of the particular logic, though one may wonder whether additional frame restrictions would help solve the problem more efficiently.
We are not aware of a way to use our translations to solve model checking more efficiently.}



As, to the best of our knowledge, most of the logics described in this paper have not been explicitly defined before, with notable exceptions such as \cite{DAGOSTINO20104273transitive,Dagostino2013S5,alberucci2009modal}, they also lack any axiomatizations and completeness theorems.
We do expect the classical methods from \cite{Kozen1983a,ladnermodcomp,Halpern1992} and others to work out in these cases as well. However it would be interesting to see if there are any unexpected situations that arise.

Given the importance of common knowledge for epistemic logic and the fact that it has been known that common knowledge can be thought of as a (greatest) fixed-point already from \cite{harman1977review,Barwise:1988:TVC:1029718.1029753}, we consider the logics that
we presented to be natural extensions of \ML.
Besides the examples given in Section \ref{sec:back}, we are interested in exploring what other natural concepts can be defined with this enlarged language.

%
%
%
\bibliographystyle{eptcs}
\bibliography{mybib}

\begin{thebibliography}{10}
\providecommand{\bibitemdeclare}[2]{}
\providecommand{\surnamestart}{}
\providecommand{\surnameend}{}
\providecommand{\urlprefix}{Available at }
\providecommand{\url}[1]{\texttt{#1}}
\providecommand{\href}[2]{\texttt{#2}}
\providecommand{\urlalt}[2]{\href{#1}{#2}}
\providecommand{\doi}[1]{doi:\urlalt{https://doi.org/#1}{#1}}
\providecommand{\eprint}[1]{arXiv:\urlalt{https://arxiv.org/abs/#1}{#1}}
\providecommand{\bibinfo}[2]{#2}

\bibitemdeclare{article}{AcetoAFI20}
\bibitem{AcetoAFI20}
\bibinfo{author}{Luca \surnamestart Aceto\surnameend}, \bibinfo{author}{Antonis
  \surnamestart Achilleos\surnameend}, \bibinfo{author}{Adrian \surnamestart
  Francalanza\surnameend} \& \bibinfo{author}{Anna \surnamestart
  Ing{\'{o}}lfsd{\'{o}}ttir\surnameend} (\bibinfo{year}{2020}):
  \emph{\bibinfo{title}{The complexity of identifying characteristic
  formulae}}.
\newblock {\slshape \bibinfo{journal}{J. Log. Algebraic Methods Program.}}
  \bibinfo{volume}{112}, p. \bibinfo{pages}{100529}.
\newblock \urlprefix\url{https://doi.org/10.1016/j.jlamp.2020.100529}.

\bibitemdeclare{article}{alberucci_facchini_2009}
\bibitem{alberucci_facchini_2009}
\bibinfo{author}{Luca \surnamestart Alberucci\surnameend} \&
  \bibinfo{author}{Alessandro \surnamestart Facchini\surnameend}
  (\bibinfo{year}{2009}): \emph{\bibinfo{title}{The modal $\mu$-calculus
  hierarchy over restricted classes of transition systems}}.
\newblock {\slshape \bibinfo{journal}{The Journal of Symbolic Logic}}
  \bibinfo{volume}{74}(\bibinfo{number}{4}), p. \bibinfo{pages}{1367–1400},
  \doi{10.2178/jsl/1254748696}.

\bibitemdeclare{article}{alberucci2009modal}
\bibitem{alberucci2009modal}
\bibinfo{author}{Luca \surnamestart Alberucci\surnameend} \&
  \bibinfo{author}{Alessandro \surnamestart Facchini\surnameend}
  (\bibinfo{year}{2009}): \emph{\bibinfo{title}{The modal $\mu$-calculus
  hierarchy over restricted classes of transition systems}}.
\newblock {\slshape \bibinfo{journal}{The Journal of Symbolic Logic}}
  \bibinfo{volume}{74}(\bibinfo{number}{4}), pp. \bibinfo{pages}{1367--1400},
  \doi{10.2178/jsl/1254748696}.

\bibitemdeclare{inproceedings}{Barwise:1988:TVC:1029718.1029753}
\bibitem{Barwise:1988:TVC:1029718.1029753}
\bibinfo{author}{Jon \surnamestart Barwise\surnameend} (\bibinfo{year}{1988}):
  \emph{\bibinfo{title}{Three views of common knowledge}}.
\newblock In: {\slshape \bibinfo{booktitle}{Proceedings of the 2nd Conference
  on Theoretical Aspects of Reasoning About Knowledge}}, \bibinfo{series}{TARK
  '88}, \bibinfo{publisher}{Morgan Kaufmann Publishers Inc.},
  \bibinfo{address}{San Francisco, CA, USA}, pp. \bibinfo{pages}{365--379}.
\newblock \urlprefix\url{http://dl.acm.org/citation.cfm?id=1029718.1029753}.

\bibitemdeclare{book}{blackburn2006handbook}
\bibitem{blackburn2006handbook}
\bibinfo{author}{Patrick \surnamestart Blackburn\surnameend},
  \bibinfo{author}{Johan \surnamestart van Benthem\surnameend} \&
  \bibinfo{author}{Frank \surnamestart Wolter\surnameend}
  (\bibinfo{year}{2006}): \emph{\bibinfo{title}{Handbook of Modal Logic}}.
\newblock {\slshape \bibinfo{series}{Studies in Logic and Practical
  Reasoning}}~\bibinfo{volume}{3}, \bibinfo{publisher}{Elsevier Science}.

\bibitemdeclare{book}{MLBlackburnRijkeVenema}
\bibitem{MLBlackburnRijkeVenema}
\bibinfo{author}{Patrick \surnamestart Blackburn\surnameend},
  \bibinfo{author}{Maarten \surnamestart de~Rijke\surnameend} \&
  \bibinfo{author}{Yde \surnamestart Venema\surnameend} (\bibinfo{year}{2001}):
  \emph{\bibinfo{title}{Modal Logic}}.
\newblock \bibinfo{series}{Cambridge Tracts in Theoretical Computer Science},
  \bibinfo{publisher}{Cambridge University Press},
  \doi{10.1017/cbo9781107050884}.

\bibitemdeclare{inproceedings}{EpistLogicRV}
\bibitem{EpistLogicRV}
\bibinfo{author}{Laura \surnamestart Bozzelli\surnameend},
  \bibinfo{author}{Bastien \surnamestart Maubert\surnameend} \&
  \bibinfo{author}{Sophie \surnamestart Pinchinat\surnameend}
  (\bibinfo{year}{2014}): \emph{\bibinfo{title}{Unifying Hyper and Epistemic
  Temporal Logic}}.
\newblock \bibinfo{volume}{abs/1409.2711}, \doi{10.1007/978-3-662-46678-0\_11}.
\newblock \eprint{1409.2711}.

\bibitemdeclare{article}{Hyperproperties}
\bibitem{Hyperproperties}
\bibinfo{author}{Michael~R. \surnamestart Clarkson\surnameend} \&
  \bibinfo{author}{Fred~B. \surnamestart Schneider\surnameend}
  (\bibinfo{year}{2010}): \emph{\bibinfo{title}{Hyperproperties}}.
\newblock {\slshape \bibinfo{journal}{Journal of Computer Security}}
  \bibinfo{volume}{18}(\bibinfo{number}{6}), p. \bibinfo{pages}{1157–1210},
  \doi{10.3233/JCS-2009-0393}.

\bibitemdeclare{article}{Dagostino2013S5}
\bibitem{Dagostino2013S5}
\bibinfo{author}{Giovanna \surnamestart D'Agostino\surnameend} \&
  \bibinfo{author}{Giacomo \surnamestart Lenzi\surnameend}
  (\bibinfo{year}{2013}): \emph{\bibinfo{title}{On modal $\mu$-calculus in S5
  and applications}}.
\newblock {\slshape \bibinfo{journal}{Fundamenta Informaticae}}
  \bibinfo{volume}{124}(\bibinfo{number}{4}), pp. \bibinfo{pages}{465--482},
  \doi{10.3233/FI-2013-844}.

\bibitemdeclare{article}{DAGOSTINO201840planar}
\bibitem{DAGOSTINO201840planar}
\bibinfo{author}{Giovanna \surnamestart D'Agostino\surnameend} \&
  \bibinfo{author}{Giacomo \surnamestart Lenzi\surnameend}
  (\bibinfo{year}{2018}): \emph{\bibinfo{title}{The $\mu$-Calculus Alternation
  Depth Hierarchy is infinite over finite planar graphs}}.
\newblock {\slshape \bibinfo{journal}{Theoretical Computer Science}}
  \bibinfo{volume}{737}, pp. \bibinfo{pages}{40--61},
  \doi{10.1016/j.tcs.2018.04.009}.
\newblock
  \urlprefix\url{https://www.sciencedirect.com/science/article/pii/S0304397518302317}.

\bibitemdeclare{article}{DAGOSTINO20104273transitive}
\bibitem{DAGOSTINO20104273transitive}
\bibinfo{author}{Giovanna \surnamestart D’Agostino\surnameend} \&
  \bibinfo{author}{Giacomo \surnamestart Lenzi\surnameend}
  (\bibinfo{year}{2010}): \emph{\bibinfo{title}{On the $\mu$-calculus over
  transitive and finite transitive frames}}.
\newblock {\slshape \bibinfo{journal}{Theoretical Computer Science}}
  \bibinfo{volume}{411}(\bibinfo{number}{50}), pp. \bibinfo{pages}{4273--4290},
  \doi{10.1016/j.tcs.2010.09.002}.
\newblock
  \urlprefix\url{https://www.sciencedirect.com/science/article/pii/S030439751000469X}.

\bibitemdeclare{article}{Dagostino2015modal}
\bibitem{Dagostino2015modal}
\bibinfo{author}{Giovanna \surnamestart D’Agostino\surnameend} \&
  \bibinfo{author}{Giacomo \surnamestart Lenzi\surnameend}
  (\bibinfo{year}{2015}): \emph{\bibinfo{title}{On the modal $\mu$-Calculus
  over finite symmetric graphs}}.
\newblock {\slshape \bibinfo{journal}{Mathematica Slovaca}}
  \bibinfo{volume}{65}(\bibinfo{number}{4}), pp. \bibinfo{pages}{731--746},
  \doi{10.1515/ms-2015-0052}.

\bibitemdeclare{article}{emerson2001model}
\bibitem{emerson2001model}
\bibinfo{author}{E~Allen \surnamestart Emerson\surnameend},
  \bibinfo{author}{Charanjit~S \surnamestart Jutla\surnameend} \&
  \bibinfo{author}{A~Prasad \surnamestart Sistla\surnameend}
  (\bibinfo{year}{2001}): \emph{\bibinfo{title}{On model checking for the
  $\mu$-calculus and its fragments}}.
\newblock {\slshape \bibinfo{journal}{Theoretical Computer Science}}
  \bibinfo{volume}{258}(\bibinfo{number}{1-2}), pp. \bibinfo{pages}{491--522},
  \doi{10.1016/S0304-3975(00)00034-7}.

\bibitemdeclare{book}{Fagin1995ReasoningAboutKnowledge}
\bibitem{Fagin1995ReasoningAboutKnowledge}
\bibinfo{author}{Ronald \surnamestart Fagin\surnameend},
  \bibinfo{author}{Joseph~Y. \surnamestart Halpern\surnameend},
  \bibinfo{author}{Yoram \surnamestart Moses\surnameend} \&
  \bibinfo{author}{Moshe~Y. \surnamestart Vardi\surnameend}
  (\bibinfo{year}{1995}): \emph{\bibinfo{title}{Reasoning About Knowledge}}.
\newblock \bibinfo{publisher}{The MIT Press},
  \doi{10.7551/mitpress/5803.001.0001}.

\bibitemdeclare{article}{fischer1979propositional}
\bibitem{fischer1979propositional}
\bibinfo{author}{Michael~J. \surnamestart Fischer\surnameend} \&
  \bibinfo{author}{Richard~E. \surnamestart Ladner\surnameend}
  (\bibinfo{year}{1979}): \emph{\bibinfo{title}{Propositional dynamic logic of
  regular programs}}.
\newblock {\slshape \bibinfo{journal}{Journal of computer and system sciences}}
  \bibinfo{volume}{18}(\bibinfo{number}{2}), pp. \bibinfo{pages}{194--211},
  \doi{10.1016/0022-0000(79)90046-1}.

\bibitemdeclare{article}{fitting1972tableau}
\bibitem{fitting1972tableau}
\bibinfo{author}{Melvin \surnamestart Fitting\surnameend}
  (\bibinfo{year}{1972}): \emph{\bibinfo{title}{Tableau methods of proof for
  modal logics.}}
\newblock {\slshape \bibinfo{journal}{Notre Dame Journal of Formal Logic}}
  \bibinfo{volume}{13}(\bibinfo{number}{2}), pp. \bibinfo{pages}{237--247},
  \doi{10.1305/ndjfl/1093894722}.

\bibitemdeclare{article}{Halpern1992}
\bibitem{Halpern1992}
\bibinfo{author}{Joseph~Y. \surnamestart Halpern\surnameend} \&
  \bibinfo{author}{Yoram \surnamestart Moses\surnameend}
  (\bibinfo{year}{1992}): \emph{\bibinfo{title}{A guide to completeness and
  complexity for modal logics of knowledge and belief}}.
\newblock {\slshape \bibinfo{journal}{Artificial Intelligence}}
  \bibinfo{volume}{54}(\bibinfo{number}{3}), pp. \bibinfo{pages}{319--379},
  \doi{10.1016/0004-3702(92)90049-4}.

\bibitemdeclare{article}{Halpern2007Characterizing}
\bibitem{Halpern2007Characterizing}
\bibinfo{author}{Joseph~Y. \surnamestart Halpern\surnameend} \&
  \bibinfo{author}{Leandro~Chaves \surnamestart R{\^e}go\surnameend}
  (\bibinfo{year}{2007}): \emph{\bibinfo{title}{Characterizing the {NP-PSPACE}
  gap in the satisfiability problem for modal logic}}.
\newblock {\slshape \bibinfo{journal}{Journal of Logic and Computation}}
  \bibinfo{volume}{17}(\bibinfo{number}{4}), pp. \bibinfo{pages}{795--806},
  \doi{10.1093/logcom/exm029}.

\bibitemdeclare{article}{harman1977review}
\bibitem{harman1977review}
\bibinfo{author}{Gilbert \surnamestart Harman\surnameend}
  (\bibinfo{year}{1977}): \emph{\bibinfo{title}{Review of linguistic behavior
  by {Jonathan Bennett}}}.
\newblock {\slshape \bibinfo{journal}{Language}} \bibinfo{volume}{53}, pp.
  \bibinfo{pages}{417--424}, \doi{10.1353/lan.1977.0036}.

\bibitemdeclare{article}{jurdzinski1998deciding}
\bibitem{jurdzinski1998deciding}
\bibinfo{author}{Marcin \surnamestart Jurdzi{\'n}ski\surnameend}
  (\bibinfo{year}{1998}): \emph{\bibinfo{title}{Deciding the winner in parity
  games is in {UP} $\cap$ {co-UP}}}.
\newblock {\slshape \bibinfo{journal}{Information Processing Letters}}
  \bibinfo{volume}{68}(\bibinfo{number}{3}), pp. \bibinfo{pages}{119--124},
  \doi{10.1016/S0020-0190(98)00150-1}.

\bibitemdeclare{article}{Kozen1983a}
\bibitem{Kozen1983a}
\bibinfo{author}{Dexter \surnamestart Kozen\surnameend} (\bibinfo{year}{1983}):
  \emph{\bibinfo{title}{Results on the propositional $\mu$-calculus}}.
\newblock {\slshape \bibinfo{journal}{Theoretical Computer Science}}
  \bibinfo{volume}{27}(\bibinfo{number}{3}), pp. \bibinfo{pages}{333--354},
  \doi{10.1016/0304-3975(82)90125-6}.

\bibitemdeclare{article}{ladnermodcomp}
\bibitem{ladnermodcomp}
\bibinfo{author}{Richard~E. \surnamestart Ladner\surnameend}
  (\bibinfo{year}{1977}): \emph{\bibinfo{title}{The computational complexity of
  provability in systems of modal propositional logic}}.
\newblock {\slshape \bibinfo{journal}{SIAM Journal on Computing}}
  \bibinfo{volume}{6}(\bibinfo{number}{3}), pp. \bibinfo{pages}{467--480},
  \doi{10.1137/0206033}.

\bibitemdeclare{unpublished}{CTMLR2022onlinepreprint}
\bibitem{CTMLR2022onlinepreprint}
\bibinfo{author}{Elli Anastasiadi Adrian Francalanza Anna~Ing\'olfsd\'ottir
  \surnamestart Luca~Aceto\surnameend, Antonis~Achilleos}:
  \emph{\bibinfo{title}{Complexity through Translations for Modal Logic with
  Recursion}}.
\newblock
  \urlprefix\url{http://icetcs.ru.is/movemnt/papers/ComplexityTranslationsMLRecFullGand.pdf}.

\bibitemdeclare{inproceedings}{Massacci1994}
\bibitem{Massacci1994}
\bibinfo{author}{Fabio \surnamestart Massacci\surnameend}
  (\bibinfo{year}{1994}): \emph{\bibinfo{title}{Strongly analytic tableaux for
  normal modal logics}}.
\newblock In: {\slshape \bibinfo{booktitle}{CADE}}, pp.
  \bibinfo{pages}{723--737}, \doi{10.1007/3-540-58156-1\_52}.

\bibitemdeclare{article}{nagle_thomason_1985}
\bibitem{nagle_thomason_1985}
\bibinfo{author}{Michael~C. \surnamestart Nagle\surnameend} \&
  \bibinfo{author}{S.~K. \surnamestart Thomason\surnameend}
  (\bibinfo{year}{1985}): \emph{\bibinfo{title}{The extensions of the modal
  logic {K5}}}.
\newblock {\slshape \bibinfo{journal}{Journal of Symbolic Logic}}
  \bibinfo{volume}{50}(\bibinfo{number}{1}), pp. \bibinfo{pages}{102–--109},
  \doi{10.2307/2273793}.

\bibitemdeclare{inproceedings}{Pratt_1981}
\bibitem{Pratt_1981}
\bibinfo{author}{Vaughan~R. \surnamestart Pratt\surnameend}
  (\bibinfo{year}{1981}): \emph{\bibinfo{title}{A decidable mu-calculus:
  Preliminary report}}.
\newblock In: {\slshape \bibinfo{booktitle}{22nd Annual Symposium on
  Foundations of Computer Science (SFCS 1981)}}, \bibinfo{publisher}{{IEEE}},
  \doi{10.1109/sfcs.1981.4}.

\bibitemdeclare{article}{RybakovShkatovBComp}
\bibitem{RybakovShkatovBComp}
\bibinfo{author}{Mikhail \surnamestart Rybakov\surnameend} \&
  \bibinfo{author}{Dmitry \surnamestart Shkatov\surnameend}
  (\bibinfo{year}{2018}): \emph{\bibinfo{title}{Complexity of finite-variable
  fragments of propositional modal logics of symmetric frames}}.
\newblock {\slshape \bibinfo{journal}{Logic Journal of the IGPL}}
  \bibinfo{volume}{27}(\bibinfo{number}{1}), pp. \bibinfo{pages}{60--68},
  \doi{10.1093/jigpal/jzy018}.

\bibitemdeclare{article}{zielonka1998infinite}
\bibitem{zielonka1998infinite}
\bibinfo{author}{Wieslaw \surnamestart Zielonka\surnameend}
  (\bibinfo{year}{1998}): \emph{\bibinfo{title}{Infinite games on finitely
  coloured graphs with applications to automata on infinite trees}}.
\newblock {\slshape \bibinfo{journal}{Theoretical Computer Science}}
  \bibinfo{volume}{200}(\bibinfo{number}{1-2}), pp. \bibinfo{pages}{135--183},
  \doi{10.1016/S0304-3975(98)00009-7}.

\end{thebibliography}
%

%
%
%
\end{document}